\documentclass[a4paper,UKenglish,cleveref, autoref, thm-restate]{lipics-v2019}
\usepackage{graphicx}
\usepackage{xspace}
\usepackage{tikz}
\usetikzlibrary{arrows,positioning,shapes,patterns,calc,decorations.pathreplacing}
\usepackage{subcaption}
\usepackage{hyperref}
\usepackage{todonotes}

\usepackage{spaces}

\author{Sven Linker}
{Department of Computer Science, University of Liverpool, UK}{s.linker@liverpool.ac.uk}{0000-0003-2913-7943}{}

\author{Fabio Papacchini}{Department of Computer Science, University of Liverpool, UK}{fabio.papacchini@liverpool.ac.uk}{}{supported by the EPSRC
through
grant EP/R026084
and grant EP/R026173.
}

\author{Michele Sevegnani}{School of Computing Science, University of Glasgow, UK}{michele.sevegnani@glasgow.ac.uk}{0000-0001-6773-9481}{supported by PETRAS SRF grant MAGIC (EP/S035362/1).}

\title{Analysing Spatial Properties on Neighbourhood Spaces}

\authorrunning{S. Linker, F. Papacchini and M. Sevegnani} 

\Copyright{Sven Linker, Fabio Papcchini and Michele Sevegnani} 



\category{} 

\relatedversion{} 

\supplement{}

\funding{
  This work was supported by the Engineering and Physical Sciences
Research Council, under the grant EP/N007565/1 (S4: Science of Sensor
Systems Software).
}

\nolinenumbers 

\hideLIPIcs  

\EventEditors{Javier Esparza and Daniel Kr{\'a}l'}
\EventNoEds{2}
\EventLongTitle{45th International Symposium on Mathematical Foundations of Computer Science (MFCS 2020)}
\EventShortTitle{MFCS 2020}
\EventAcronym{MFCS}
\EventYear{2020}
\EventDate{August 24--28, 2020}
\EventLocation{Prague, Czech Republic}
\EventLogo{}
\SeriesVolume{170}
\ArticleNo{61}

\begin{document}

\maketitle

\begin{abstract}
  We present a bisimulation relation for neighbourhood spaces,
  a generalisation of topological spaces. We show that this notion,
  \emph{path preserving bisimulation},
  preserves formulas of the spatial logic SLCS.
  We then use this preservation result to show that SLCS cannot express standard topological properties such as separation and
connectedness.
  Furthermore, we compare the bisimulation relation with standard modal bisimulation and
  modal bisimulation with converse on graphs and prove it coincides with the latter.
\end{abstract}

\section{Introduction}
\label{sec:intro}


The functionality of modern computer systems is increasingly affected by their spatial properties. 
%
For example, correctness and efficiency of distributed algorithms
depend on the underlying network topology, e.g., whether nodes are
reachable, or if there are disconnected components. Furthermore, for cyber-physical
systems like autonomous vehicles, spatial aspects are crucial for safe behaviour.
To reason about spatial properties, there exist a variety of
spatial logics \cite{Aiello2007} with different kinds of semantics: geometric, directional, topological,
or based on structural properties of concurrent processes \cite{Caires2001}.  
However, the analysis of such spatial logics is much less evolved than the analysis
of temporal logics like linear temporal logic \cite{Pnueli1977} or
computation tree logic \cite{Emerson1986}.

In this paper, we focus on a kind of spatial logics defined on
\emph{neighbourhood spaces} also called \emph{\v{C}ech closure spaces} \cite{Cech1966} or
pretopological spaces:
a generalisation of topological
spaces, where the closure operator is not required to be idempotent. 
In particular,
we analyse the \emph{Spatial Logic on Closure Spaces} (SLCS) introduced by Ciancia et al. \cite{Ciancia2017}.
So far, there exists a model-checking algorithm for SLCS, and it has been used for analysis in various application domains such as congestion in bike-sharing applications \cite{Ciancia2016} and
bus schedules~\cite{Ciancia2018}. An extension of SLCS with distance measuring
operators has been used to analyse medical images \cite{Banci2019,DBLP:conf/tacas/BelmonteCLM19}.
However, to the best of our knowledge, no further study of the
overall properties of SLCS has been conducted. For example, it is still an open question what its limits of
expressivity are.
To relate the structural properties of models to a logical language, we follow the standard approach of defining various notions of bisimulations~\cite{BlackburnVanBenthem07}
and studying the invariance of SLCS modalities. To that end, we follow ideas of Kurtonina
and de Rijke by extending the bisimulations to cover paths~\cite{KurtoninaDeRijke97}.  
We also employ these bisimulations to study SLCS on two important subclasses of neighbourhood spaces. The first class consists of \emph{topological spaces}, while the latter
is the class of \emph{quasi-discrete spaces}, which can be thought of as (possibly infinite) graphs.
These classes are non-disjoint, and neither is a subclass of the other. Furthermore, all finite spaces
are quasi-discrete.


The investigation of this paper was inspired by recent work of
Baryshnikov and Ghrist~\cite{Baryshnikov2009} on a topological approach
to the \emph{target counting problem} in sensor networks, the
computational task of determining the total number of targets in a
region by aggregating the individual counts of each sensor without
recording any target identities nor any positional information. Its
mathematical formulation depends on having sensor readings over a
continuum field of sensors. However, any implementation must occur
over a discrete collection of sensors in a given network. This
introduces some limitations as several studies have highlighted
\cite{8064025,doi:10.1098/rspa.2019.0278}, in particular it is almost impossible
to predict
the accuracy of the results a given discretisation  yields. This
shows the need for general notions to rigorously study how properties of
interests are preserved across different kind of spaces and provides
motivation for this work.

Our contributions in this paper are as follows. 
\begin{itemize}
\item Definition of bisimulations between neighbourhood models;
\item proof that bisimilar points satisfy the same SLCS formulas;
\item use of the defined bisimulations to study expressivity of SLCS; and
\item comparison of the introduced notions with bisimulations on graphs treated as neighbourhood spaces.
\end{itemize}


Our article is organised as follows. We begin in Sect. \ref{sec:closurespaces} by presenting some preliminary background on neighbourhood spaces.
Sect. \ref{sec:slcs} introduces the  main bisimulation relation: path preserving bisimulation.
In Sect. \ref{sec:quasi}, 
we study the properties of this bisimulation on quasi-discrete spaces.
Related work is presented in Sect. \ref{sec:related} and we conclude our work in Sect.~\ref{sec:conclusion}. The full proofs have been moved to the appendix.


\section{Neighbourhood Spaces}
\label{sec:closurespaces}

In this section we recall the notions of neighbourhood spaces and some related results from general topology we will use in this paper.
Our main reference is~\cite{Cech1966}. For additional general results on these topics
and for the proofs of the results reported here, we refer the reader to this
source.

\begin{definition}[Filter]
  Given a set \(X\), a \emph{filter} \(F\) on \(X\) is a subset of \(\powset{X}\), such that \(F\) is closed under non-empty intersections,
  whenever \(Y \in F\) and \(Y \subseteq Z\), then also \(Z \in F\), and finally  \(\emptyset \not\in F\). For a
  set \(A\subseteq X\), the filter generated by \(A\) is written as \(\genfilter{A}\).
\end{definition}

\begin{definition}[Neighbourhood Space]
  \label{def:nbhd-space}
  Let \(X\) be a set together with \(\nbhdSing{} \subseteq \powset{\powset{X}}\) given by \(\nbhdSing{} = \{\nbhd{}{x} \mid x \in X \}\), where every \(\nbhd{}{x}\) is a filter on \(X\) and
  \(x \in \bigcap_{N \in \nbhd{}{x}} N\).
We call \(\nbhdSing{}\) a \emph{neighbourhood system} on \(X\), and \(\closurespace = (X, \nbhdSing{})\) a \emph{neighbourhood space}.
For every set \(A\subseteq X\), we have the (unique) \emph{interior} and \emph{closure} operators defined
as follows.
\begin{align*}
  \interior{\nbhdSing{}}{A} & = \{ x \in A \mid A \in \nbhd{}{x}\}&
  \closure{\nbhdSing{}}{A} & = \{x \in X \mid \forall N \in \nbhd{}{x} \colon A \cap N \neq \emptyset\}
\end{align*}
An element \(x \in X\)
has a \emph{minimal neighbourhood} if there exists
\(N \in \nbhd{}{x}\)
such that $N \subseteq N'$ for any  neighbourhood
$N' \in \nbhd{}{x}$. We use $\minNbhd{x}$ to refer to the minimal
neighbourhood of $x$. If each element \(x \in X\)
has a minimal neighbourhood, then we call \(\closurespace\)
\emph{quasi-discrete}.  Finally, if for every element \(x \in X\)
and any neighbourhood \(N \in \nbhd{}{x}\),
there is a neighbourhood \(M \in \nbhd{}{x}\),
such that for every \(y \in M\),
we have also that \(N \in \nbhd{}{y}\),
then \(\closurespace\) is \emph{topological}.
\end{definition}



\begin{proposition}[Closure Operator (\cite{Cech1966} 14 A.1, 14 B.11, 15 A.1, 15 A.2, 26 A.1, 26 A.9)]
  \label{prop:closure}
  For any neighbourhood space \(\closurespace = (X,\nbhdSing{})\),  the  closure operator
  \(\closureSing{}\) as induced
  by \(\nbhdSing{}\) satisfies the following properties:
  \begin{enumerate}
  \item \(\closure{}{\emptyset} = \emptyset\) \label{cls:empty}
  \item \(A \subseteq \closure{}{A}\) \label{cls:mon}
  \item \(\closure{}{A \cup B} = \closure{}{A} \cup \closure{}{B}\) \label{cls:union}
  \item If \(\closurespace\) is quasi-discrete then,  for any set \(A \subseteq X\),
    \(\closure{}{A} = \bigcup_{a \in A}\closure{}{\{a\}}\). \label{cls:quasi-discrete}
  \item   If \(\closurespace\) is topological, then  
    for any set \(A \subseteq X\), \(\closure{}{A} = \closure{}{\closure{}{A}}\).
  \end{enumerate}
\end{proposition}

In the work of \v{C}ech \cite{Cech1966}, the properties of Proposition~\ref{prop:closure} are used to define
closure operators, and the equivalences with the corresponding properties
of the neighbourhood systems are shown in several theorems. However, since
we will use neighbourhoods as the primary entities in the spaces, we choose
to demote the closure operators to be derived. 

\begin{definition}[Connectedness (\cite{Cech1966} 20 B.1)]
  \label{def:connected}
  Let \(\closurespace = (X,\nbhdSing{})\) be a neighbourhood space. Two subsets \(U\) and
  \(V\) of \(X\) are \emph{semi-separated}, if \(\closure{}{U} \cap V = U \cap \closure{}{V} = \emptyset\).
  A subset \(U\) of \(\closurespace\) is \emph{connected}, if it is not the union of
  two non-empty, semi-separated sets. The space \(\closurespace\) is connected, if \(X\)
  is connected.
\end{definition}

We also introduce a special kind of neighbourhood space, employed with a linear order.
\begin{definition}[Index Space]
  \label{def:index_space}
  If \((I, \nbhdSing{})\) is a connected neighbourhood space and \({\leq} \subseteq I \times I \) a
  linear order on \(I\) with the bottom element \(0 \in I\), then we call \(\indexclspace = (I, \nbhdSing{}, \leq, 0)\) an \emph{index space}.
\end{definition}

In the following sections, we will often use the concept of continuous function. Generally, we will
use the notation \(f[A]\) for  the image of a set \(A \subseteq X\) under a function \(f \colon X \to Y\) 
. Similarly, \(f^{-1}[B]\) denotes
the preimage of a set \(B \subseteq Y\)
.

\begin{definition}[Continuous Function (\cite{Cech1966} 16 A.4)]
  \label{def:continuous}
  Let \(\closurespace_i = (X_i, \nbhdSing{i})\) for \(i \in \{1,2\}\) be two neighbourhood spaces. A function
  \(f \colon X_1 \to X_2\) is \emph{continuous}, if for every \(x_1 \in X_1\) and every \(N_2 \in \nbhd{2}{f(x_1)}\),
  there is a  \(N_1 \in \nbhd{1}{x_1}\) such that  \(f[N_1] \subseteq N_2\).
  Equivalently, since the neighbourhood system of \(x_1\) is upward closed, for every neighbourhood \(N_2 \in \nbhd{2}{f(x_1)}\), \(f^{-1}[N_2] \in \nbhd{1}{x_1}\).
  We will also write
  \(f \colon \closurespace_1 \to \closurespace_2\).
\end{definition}

Observe that this coincides with the well-known definition of continuous functions on topological spaces. An important connection between connected sets and continuous functions
is that the image of a connected set is connected.

\begin{lemma}[Connectedness and Continuity (\cite{Cech1966} 20 B.13)]
  \label{lem:continuous_connected}
  Let \(f \colon \closurespace_1 \to \closurespace_2\) be continuous. If
  a subset \(X\) of \(\closurespace_1\) is connected, then \(f[X]\) is connected.
\end{lemma}

Following Ciancia et al.~\cite{Ciancia2017}, we extend the typical
notion of a topological path to neighbourhood spaces. 

\begin{definition}[Path]
  \label{def:path}
  For an index space \(\indexclspace\) and a neighbourhood space \(\closurespace\), a continuous
  function \(p \colon \indexclspace \to \closurespace\) is a \emph{path on \(\closurespace\)}.
  If \(p(0) = x\), we will also write \(\pathdef{}{x}\) to denote \emph{a path starting in \(x\)}.
\end{definition}

This definition includes both quasi-discrete paths and topological paths as given
by Ciancia et al. \cite{Ciancia2017}. For example, two typical index spaces are
\(\indexclspace = (\R, \nbhdSing{\R},\leq, 0)\)
with the standard topology based on open intervals,   and
\(\indexclspace = (\N, \nbhdSing{\N}, \leq, 0)\), where \(\nbhdSing{\N}\) is given by
the quasi-discrete  neighbourhood system induced by the successor relation. That is, the minimal
neighbourhood of each point \(n\) is given by \(\{n, n+1\}\).
Furthermore, observe that by the definition of index spaces and Lemma~\ref{lem:continuous_connected},
the image of a path is connected.

We now
 present spatial models based on neighbourhood spaces and, based on that, the
syntax and semantics of SLCS. For the rest of the paper, we let
\(\propositions\) be a fixed  denumerable set of
propositional atoms. 

\begin{definition}[Neighbourhood Model]
  Let \(\closurespace = (X,\nbhdSing{})\) 
  be a neighbourhood space, \(\indexclspace\) an index space,
  and let \(\valSing{} \colon X \to \powset{\propositions}\) be
  a valuation. Then \(\topomodel = (\closurespace, \indexclspace,
  \valSing{})\) is a
  \emph{neighbourhood model}. We will also
  write \(\topomodel = (X,\nbhdSing{}, \valSing{})\) to
  denote neighbourhood models, if the index space is clear from
  the context.
\end{definition}

We lift all suitable previous definitions  to neighbourhood models
in the obvious ways. For example,
we will speak of continuous functions between the underlying
spaces of two models as continuous functions between the models.

\begin{definition}[Syntax of SLCS]
  \label{def:syntax}
  \begin{align*}
    \varphi \defeq p \mid \ltrue \mid \lnot \varphi \mid \varphi \land \varphi \mid \near \varphi \mid \varphi \reachable \varphi \mid \varphi \propagate \varphi
  \end{align*}
  \(\near\) is read as \emph{near}, \(\reachable\) is read as \emph{reachable from}, and \(\propagate\) is read as \emph{propagates to}.
\end{definition}
The intuition behind the modalities is as follows. A point satisfies \(\near \varphi\),
if it is contained in the closure of the set of points satisfying \(\varphi\). Hence, even if
it does not satisfy \(\varphi\) itself, it is close to a point that does. A point \(x\) is satisfying
\(\varphi \reachable \psi\) if there is a point \(y\) satisfying \(\psi\) such that \(x\) is
reachable from \(y\) via a path where every point on this path between \(x\) and \(y\) satisfies
\(\varphi\). Propagation is in a sense the converse modality, i.e., if there is a point
\(y\) satisfying \(\psi\) such that there is a path starting in \(x\) and reaching \(y\)
at some index, and all points in between satisfy \(\varphi\), then \(x\) satisfies \(\varphi \propagate \psi\). This intuition is formalised in the following semantics.

\begin{definition}[Path Semantics of SLCS]
  \label{def:path-semantics}
  Let \(\topomodel = (\closurespace, \indexclspace, \valSing{})\) be a  neighbourhood model and \(x \in \closurespace\). The semantics of
  SLCS with respect to \(\topomodel\) is defined inductively as follows.\footnote{The original definition of the path semantics by Ciancia et al. \cite{Ciancia2017} differs from our presentation.
    This is due to a change in their definition of the closure operator. In particular,
    they define the closure on quasi-discrete spaces, i.e.,  with respect to a given relation \(R\) as
    \(\closure{R}{A} = A \cup \{x \in X \mid \exists a \in A \colon (a,x) \in R\}\). Our definition  yields \(\closure{R}{A} = A \cup \{x \in X \mid \exists a \in A \colon (x,a) \in R\}\) (see the discussion at the end of this section), which is more in line with  other literature \cite{Cech1966,Galton2003}. However, this only changes whether \(\near\) can be considered the one-step counterpart of \(\reachable\) or of \(\propagate\).
  }
  \begin{align*}
    \topomodel, x &\models \ltrue & &\text{for all } \topomodel \text{ and } x\\
    \topomodel, x &\models p & \iff & p \in \val{}{x}\\
    \topomodel, x & \models \lnot \varphi & \iff & \text{not } \topomodel, x \models \varphi\\
    \topomodel, x & \models \varphi \land \psi & \iff &  \topomodel, x \models \varphi \text{ and } \topomodel, x \models \psi\\
    \topomodel, x & \models \near \varphi & \iff & x \in \closure{}{\{y \mid \topomodel, y \models \varphi\}} \\
    \topomodel, x & \models \varphi \reachable \psi& \iff &\text{there are } y, n \text{ and }  \pathdef{}{y} 
                                                            \text{ such that } \pathAt{}{n} = x \text{ and } \model, y \models \psi\\
    & & & \text{ and for all } 0 < i \leq n \colon \model, \pathAt{}{i} \models \varphi\\
    \topomodel, x & \models \varphi \propagate \psi & \iff &  \text{there are } \pathdef{}{x} \text{ and } n \text{ such that }  \topomodel, \pathAt{}{n} \models \psi  \\
    &&& \text{ and for all } i \colon 0 \leq i < n \colon  \topomodel, \pathAt{}{i} \models \varphi  
  \end{align*}
\end{definition}

 Ciancia et al. base SLCS on a slightly different set of operators~\cite{Ciancia2017}. In particular,
 they employ a modality \(\surrounded\), where
\(\varphi \surrounded \psi\) expresses
that the current point is within a set satisfying \(\varphi\) that
is \emph{surrounded} by a set of points satisfying \(\psi\). However, we chose to have a more
symmetric set of operators, and thus use \(\reachable\) instead. This is not problematic, since
\(\surrounded\) can be expressed by the following equivalence: \((\varphi \surrounded \psi)
\leftrightarrow (\varphi \land \lnot ( \varphi \reachable \lnot (\varphi \lor \psi)))\).

Let \(\model = (\closurespace, \indexclspace, \valSing{})\) be a model, and $\pathSing{}$
a path $\pathdef{}{x}$ in \(\model\). For \(n,m \in \indexclspace\) and \(n < m\),
we use \((n,m)\) as notation for the set
\(\{ i  \mid n < i < m\}\), similar to the usual notation of open
intervals on the indexspace \(\indexclspace\). For such an interval \((n,m)\)
and an SLCS formula \(\varphi\), we use the following abbreviation
to denote the \emph{satisfaction of \(\varphi\) within \((n,m)\)}:
\[
\model, p, (n,m) \models \varphi \text{ iff for all } i \text{ with }n < i <
m \text{ we have } \model, \pathSing{}(i) \models \varphi\ .
\]


With this notation, the semantics of  $\reachable$ and
$\propagate$ read as follows.
\begin{align*}
  \model,x \models \varphi \reachable \psi &\text{ iff }\exists
                                             \pathdef{}{y} \text{
                                             and $n$ s.t. }
                                             p(n)=x, \model,
                                             y\models \psi,
                                             \model,x \models \varphi, \\ &\text{ and } \model,
                                             p, (0,n) \models \varphi\\
  \model,x \models \varphi \propagate \psi &\text{ iff }\exists
                                             \pathdef{}{x} \text{
                                             and $n$
                                             s.t. }\model,
                                             p(n)\models \psi,
                                             \model,x \models \varphi,\text{ and } \model,
                                             p, (0,n) \models \varphi
\end{align*}

While we are able to define SLCS for the setting of general neighbourhood models, we will often
restrict our attention to one of the following two special cases:
quasi-discrete and topological models. They are defined as follows.

\begin{definition}[Quasi-Discrete and Topological  Models]
  \label{def:special-models}
  Let \(\closurespace\) be a quasi-discrete neighbourhood space,
  and \(\indexclspace_{\N}=(\N, \nbhdSing{\N}, \leq, 0)\) be the index space defined by the 
  natural numbers. Then a model \(\model = (\closurespace, \indexclspace_{\N}, \valSing{})\)
  based on these spaces is a \emph{quasi-discrete neighbourhood model}. Similarly,
  if \(\closurespace\) is topological, and \(\indexclspace_{\R} = (\R, \nbhdSing{\R}, \leq, 0)\)
  is the index space defined by the real numbers, and the topology based on all
  open intervals as well as the standard ordering of the reals, a model
  \(\model = (\closurespace, \indexclspace_{\R}, \valSing{})\) is a
  \emph{topological neighbourhood model}.
\end{definition}

Hence, whenever we refer to a model as quasi-discrete, we fix the
index space to the natural numbers, and similarly, whenever
a model is topological, we only allow for topological paths.
Observe that every quasi-discrete space can be described as a
(possibly infinite) graph structure. For a quasi-discrete space $(X,\nbhdSing{})$ the
induced edge relation $R \subseteq X\times X$ is defined as
$\{(x,y) \mid y \in \minNbhd{x}\}$. This results in the closure
operator  being defined on points of a quasi-discrete space as
$\closureSing{}(x) = \{y \in X \mid x \in \minNbhd{y}\}$. Furthermore,
as $x\in \minNbhd{x}$ for any $x\in X$, it follows that $R$ is
reflexive (as also shown in~\cite{Cech1966} 26 A.2). On the other hand,
every graph \(G = (V,R)\) (where \(R\subseteq V\times V \) is not
necessarily reflexive) induces a quasi-discrete space, by setting
the minimal neighbourhood of a vertex \(x\in V\) to be
\(\minNbhd{x} = \{x\} \cup \{y \mid (x,y) \in R\}\). Whenever we
depict quasi-discrete models as graphs, we will omit the implicit
loops on nodes.

Of course, there are neighbourhood spaces that are both quasi-discrete
\emph{and} topological. This is the case if the edge relation
of the graph representation of a quasi-discrete space is transitive (see \cite{Cech1966}, Theorem 26 A.2).
In particular, fully connected bidirectional
graphs are also topological, if considered as neighbourhood spaces. For such spaces, we
have to restrict ourselves to treat them either as topological or as quasi-discrete.


\section{Bisimulations for Neighbourhood Spaces}
\label{sec:slcs}
In this section we define two notions of bisimulation for neighbourhood spaces: \emph{neighbourhood bisimulation} and \emph{path preserving bisimulation}. We will then use them to study the preservation of SLCS formulas across models and thus the expressivity of SLCS. 

\begin{definition}[Neighbourhood Bisimulation]
  \label{def:nbhd_bisim}
  Let \((X_1, \nbhdSing{1}, \valSing{1})\) and \((X_2, \nbhdSing{2}, \valSing{2})\) be two neighbourhood models over the same index space, 
  and \(x_1 \in X_1\), \(x_2 \in X_2\) two points of the respective models.
  A relation \(Z_{\nbhdSing{}} \subseteq X_1 \times X_2\) with \(x_1 Z_{\nbhdSing{}} x_2\) is a \emph{neighbourhood bisimulation}
 of \(x_1\) and \(x_2\), if we have
  \begin{description}
  \item[\itatomic] \(p \in \val{1}{x_1}\)
    if, and only if, \(p \in \val{2}{x_2}\)
    for all \(p \in \propositions\) \label{nbhdbisim:sat}
  \item[\itforthnbhd] for every neighbourhood \(N_2
    \in \nbhd{2}{x_2}\), there is a neighbourhood \(N_1 \in
    \nbhd{1}{x_1}\) such that for all \(y_1 \in N_1\), there is a
    \(y_2 \in N_2\) with \(y_1 Z_{\nbhdSing{}} y_2\) \label{nbhdbisim:forth}
  \item[\itbacknbhd] for every neighbourhood
    \(N_1 \in \nbhd{1}{x_1}\),
    there is a neighbourhood \(N_2 \in \nbhd{2}{x_2} \)
    such that for all \(y_2 \in N_2\),
    there is a \(y_1 \in N_1\)
    with \(y_1 Z_{\nbhdSing{}} y_2\) \label{nbhdbisim:back}
  \end{description}
  Two models $\topomodel_1$ and $\topomodel_2$ are \emph{neighbourhood
    bisimilar} at $x_1$ and $x_2$, 
  if there
  is a neighbourhood bisimulation \(Z_{\nbhdSing{}}\) such that \(x_1 Z_{\nbhdSing{}} x_2\).
\end{definition}

We can prove that SLCS formulas using only the ``near'' modality are
invariant under neighbourhood bisimulation. While we do not present a
separate theorem for this fact due to space reasons,  its proof can be
extracted from the corresponding induction step of the proof of
Theorem~\ref{thm:slcs-invariant}.

\begin{example}
  Let
  \(\model_\R = ((\R, \nbhdSing{\R}), \indexclspace_\R,
  \valSing{\R})\)
  be a topological neighbourhood model, where the underlying space is
  given by the usual topology on the real numbers, and
  \(\val{\R}{s} = \{a\}\)
  for all \(s \in (-1,1)\)
  and \(\val{\R}{s} = \emptyset\)
  otherwise.  Furthermore, let
  \(\model_2 = ((\{x,y\}, \nbhdSing{2}), \indexclspace_\R,
  \valSing{2})\)
  be a topological model where \(\nbhdSing{2}\)
  is the discrete topology on the set \(\{x,y\}\) (i.e., \(\minNbhd{x} = \{x\}\) and
  \(\minNbhd{y} = \{y\}\)),
  \(\val{2}{x} = \{a\}\),
  and \(\val{2}{y} = \emptyset\).
  Then the relation \(Z_{\nbhdSing{}}\),
  given by \(s Z_{\nbhdSing{}} x\)
  for all \(s \in (-1,1)\),
  is a neighbourhood bisimulation between any point \(s \in (-1,1)\)
  and \(x\).
  
  Observe that it is not total, and in particular, there cannot be
  a total neighbourhood bisimulation between these two spaces: If there was,
  it would need to relate \(1\) to \(y\), since neither satisfies any proposition, and
   \(y\) is the only such point in \(\model_2\). However, consider the
  neighbourhood \(\{y\} \in\nbhd{2}{y}\). Every neighbourhood of \(1\) contains
  a point \(s < 1\), which is not in relation with \(y\). Hence, there is no
  neighbourhood \(N\) of \(1\) such that every element of \(N\) is in relation
  with an element of \(\{y\}\).
\end{example}

In the preceeding example, all points that are related by
\(Z_{\nbhdSing{}}\)
indeed satsify the same formulas using only \(\near\),
in this case Boolean combinations of the formulas \(\near a\)
and \(\lnot \near \lnot a\)
(or equivalent formulas).  However,
\(\model_\R, 0 \models a \propagate \lnot a\),
while \(\model_2,x \not\models a \propagate \lnot a\).
To ensure the preservation of formulas using the path modalities
\(\propagate\)
and \(\reachable\),
we strengthen our notion of bisimulation following ideas of Kurtonina
and de Rijke \cite{KurtoninaDeRijke97}. Specifically, we not only need
points in the two models to be related, but also intervals over
paths. This is achieved by introducing two relations $Z_1$ and $Z_2$,
the former relating path intervals from the first model to the second,
and the latter the other way around. The resulting bisimulation is
based on a triple or relations $(Z_{\nbhdSing{}}, Z_1, Z_2)$ and
defined as follows.

\begin{definition}[Path Preserving Bisimulation]
  \label{def:path_preserving_bisim}
  Let \(\model_1 = ((X_1, \nbhdSing{1}), \indexclspace, \valSing{1})\)
  and \(\model_2 = ((X_2, \nbhdSing{2}), \indexclspace, \valSing{2})\)
  be two neighbourhood models over the same index space \(\indexclspace\),
  and $\mathcal P$ and $\mathcal Q$ sets
  of all possible paths on $\model_1$ and $\model_2$, respectively.  A
  \emph{path preserving bisimulation} between $\model_1$ and
  $\model_2$ is triple $(Z_{\nbhdSing{}}, Z_1, Z_2)$, where
  $Z_{\nbhdSing{}} \subseteq X_1 \times X_2$, $Z_1$ a relation between
    $\mathcal P  \times \indexclspace$ and $\mathcal Q
  \times \indexclspace $, and $Z_2$ a relation between
    $\mathcal Q  \times \indexclspace$ and $\mathcal P
  \times \indexclspace$ s.t.
  $Z_{\nbhdSing{}} \neq \emptyset$ and the following holds for all
  $x_1\in X_1$, $x_2\in X_2$, $(p,n)\in \mathcal P
  \times \indexclspace$ and $(q,m) \in \mathcal Q
  \times \indexclspace$. 
  \begin{enumerate}
  \item if $x_1\mathrel{Z_{\nbhdSing{}}}x_2$, then
    $Z_{\nbhdSing{}}$ is a neighbourhood bisimulation; \label{ppb:nbhd}
  \item if $x_1\mathrel{Z_{\nbhdSing{}}}x_2$,
    $\pathdef{}{x_1}$ and $n \neq 0$, then there exists
    $q \colon x_2 \rightsquigarrow \infty$ and $m$
    s.t. $p(n)\mathrel{Z_{\nbhdSing{}}}q(m)$ and
    $(p,n)\mathrel{Z_1}(q,m)$; \label{ppb:ex_path_forth1}
  \item if $x_1\mathrel{Z_{\nbhdSing{}}}x_2$,
    $\pathdef{}{y_1}$ with $p(n)=x_1$ and $n \neq 0$, then there exists
    $q \colon y_2 \rightsquigarrow \infty$ and $m$ with $q(m)=x_2$
    s.t. $p(0)\mathrel{Z_{\nbhdSing{}}}q(0)$ and $(p,n)\mathrel{Z_1}(q,m)$;
  \item if $(p,n)\mathrel{Z_1}(q,m)$ and there exists $k_q \in \indexclspace$ with $0 <
    k_q < m$, then there exists $k_p\in \indexclspace$ with $0 < k_p < n$ s.t. $p(k_p)\mathrel{Z_{\nbhdSing{}}}q(k_q)$; \label{ppb:ex_path_forth2}
  \item if $x_1\mathrel{Z_{\nbhdSing{}}}x_2$,
    $q \colon x_2 \rightsquigarrow \infty$ and $m \neq 0$, then there exists
    $p \colon x_1 \rightsquigarrow \infty$ and $n$
    s.t. $p(n)\mathrel{Z_{\nbhdSing{}}}q(m)$ and
    $(q,m)\mathrel{Z_2}(p,n)$;
  \item if $x_1\mathrel{Z_{\nbhdSing{}}}x_2$, $q \colon y_2
    \rightsquigarrow \infty$ with $q(m)=x_2$ and $m \neq 0$, then there exists
    $p \colon y_1 \rightsquigarrow \infty$ and $n$ with $p(n)=x_1$
    s.t. $p(0)\mathrel{Z_{\nbhdSing{}}}q(0)$ and
    $(q,m)\mathrel{Z_2}(p,n)$; and
  \item if $(q,m)\mathrel{Z_2}(p,n)$ and there exists $k_p \in \indexclspace$ with $0 <
    k_p < n$, then there exists $k_q\in \indexclspace$ with $0 < k_q < m$
    s.t. $p(k_p)\mathrel{Z_{\nbhdSing{}}}q(k_q)$.
  \end{enumerate}
\end{definition}

It is straightforward to show that for three models \(\model_1\), \(\model_2\),
and \(\model_3\) over the same index space \(\indexclspace\), whenever
there is a path preserving bisimulation 
between \(x_1 \in \model_1\) and \(x_2 \in \model_2\), and there is a path
preserving bisimulation between \(x_2 \) and \(x_3 \in \model_3\), then
 there is also a path preserving bisimulation between \(x_1\) and \(x_3\).

Before we show that the truth of all SLCS formulas is preserved under path preserving
bisimulation, we first present the following technical lemma.
\begin{lemma}
  \label{lem:auxilliary_path_preservation}
  Let \((Z_{\nbhdSing{}}, Z_1, Z_2)\) be a path preserving bisimulation between \(\model_1\) and
  \(\model_2\), and \(\varphi\) be an SLCS formula that is invariant under
  neighbourhood bisimulation, i.e., for any \(x_1 \in \model_1\) and \(x_2 \in \model_2\)  
  with \(x_1 Z_{\nbhdSing{}} x_2\), we have
  \(\model_1, x_1 \models \varphi\) if, and only if, \(\model_2, x_2 \models \varphi\).
  For two paths \(p\) and \(q\) with \((p,n)\mathrel{Z_1} (q,m)\), we have
  \(\model_1, p,(0,n) \models \varphi\) implies \(\model_2,q,(0,m) \models \varphi\). Additionally,
  if \((q,m) Z_2 (p,n)\) then \(\model_2, q,(0,m) \models \varphi\) implies \(\model_1, p,(0,n) \models
  \varphi\).
\end{lemma}

\begin{theorem}
  \label{thm:slcs-invariant}
  If \((Z_{\nbhdSing{}},Z_1,Z_2)\) is a path preserving bisimulation between
  \(\model_1\) and \(\model_2\)
  with \(x_1 Z_{\nbhdSing{}} x_2\), then \(\model_1, x_1 \models \varphi\) if, and only if,
  \(\model_2, x_2 \models \varphi\) for every formula \(\varphi\) of SLCS.
\end{theorem}
\begin{proof}
  We proceed
  by induction on the length of formulas. The induction base and the cases for the Boolean
  operators are as usual. For the \emph{near} modality, the induction step consists basically
  of a straightforward application of the definitions.  We provide a sketch for
  the preservation of \emph{propagate}. The case for \emph{reachable} is analogous.

  So let \(\model_1,x_1 \models \varphi \propagate \psi\). That is, there is a path
  \(p\) starting in \(x_1\) and visiting a point satisfying \(\psi\) at the index \(n\),
  where all points in between satisfy \(\varphi\). By the bisimulation property (Def. \ref{def:path_preserving_bisim}~(\ref{ppb:ex_path_forth1})),
  there is a path \(q\) starting in \(x_2\) that visits, at \(m\), a point that is bisimilar to \(p(n)\), and for all indices between \(0\) and \(m\), there are bisimilar points on \(p\)
  as well. Hence, by the induction hypothesis and Lemma~\ref{lem:auxilliary_path_preservation},
  \(q\) is a witness that \(\model_2, x_2 \models \varphi \propagate \psi\).
  The other direction is similar, using the second case of Lemma~\ref{lem:auxilliary_path_preservation}. 
   \end{proof}

   Note that we do not show that logical equivalence of two points implies that they
   are bisimilar (cf. Sect.~\ref{sec:conclusion}). Now that we have a suitable notion of bisimilarity, we can use it
to analyse whether SLCS is able to capture spatial properties. As an example,
we show that SLCS 
is neither capable of expressing standard topological separation axioms
nor the connectedness of a  model.
\begin{definition}[Separation Properties]
  \label{def:t0-separated}
  Let \(\closurespace\) be a neighbourhood space. If for every two
  points \(x, y \in \closurespace\) we have that \(y \in \closure{}{\{x\}}\) and
  \(x \in \closure{}{\{y\}}\) implies \(x = y\), then \(\closurespace\) is
  \emph{\(T_0\)-separated}.
  If \(\{x\} \cap \closure{}{y} = \closure{}{x} \cap \{y\} = \emptyset\) for all
  distinct \(x\) and \(y\), then  \(\closurespace\) is  \(T_1\)-separated.\footnote{\v{C}ech calls such spaces \emph{feebly semi-separated} and \emph{semi-separated}, respectively, \cite{Cech1966}, but
    the name \(T_0\) and \(T_1\)  for these properties are standard in topology.}
  We call a neighbourhood model  \(T_i\)-separated, if its underlying space is
  \(T_i\)-separated for \(i \in \{0,1\}\).
\end{definition}

   \begin{proposition}
  \label{prop:t0-separation-inexpressible}
  There is no formula of SLCS   expressing \(T_0\) separation.
\end{proposition}
\begin{proof}
  Consider the quasi-discrete models
  \(\topomodel_1\) and \(\topomodel_2\) in Fig.~\ref{fig:t0-separation},
  and the relation \(Z_{\nbhdSing{}}\) given by \(x_i Z_{\nbhdSing{}} y_i\) and \(x_0 Z_{\nbhdSing{}} y_0^\prime\),
  where \(Z_1\) is defined by \((p,n) Z_1 (q,n)\) iff \(p(0) Z_{\nbhdSing{}} q(0)\) and
  \begin{align*}
    p(i) = x_0 &\Leftrightarrow  q(i) \in \{y_0, y_0^\prime\} \enspace, \\
    p(i) = x_1& \Leftrightarrow q(i) = y_1 \enspace, \\
    p(i) = x_2&\Leftrightarrow q(i) = y_2 \enspace .
  \end{align*}
  The relation \(Z_2\) is then given by \(Z_2 = Z_1^{-1}\).
  Then the triple of these three relations  is a path preserving bisimulation
  between \(\topomodel_1\) and
  \(\topomodel_2\). For example, 
  consider the minimal neighbourhood \(\minNbhd{x_1} =\{x_1, x_2\}\) of \(x_1\).
  Then choose \(\minNbhd{y_1} = \{y_1,y_2\}\) as a neighbourhood of \(y_1\). For every element of \(\minNbhd{y_1}\),
  there is an element in \(\minNbhd{x_1}\), such that the elements are bisimilar.
  The other neighbourhoods
  can be checked similarly. So, all points in \(\topomodel_1\) and \(\topomodel_2\) satisfy
  the same set formulas of SLCS by Theorem~\ref{thm:slcs-invariant}. But it is
  also easy to check that \(\topomodel_1\) is \(T_0\)-separated, while \(\topomodel_2\) is not. Hence
   no formula of SLCS expresses \(T_0\)-separation.
\end{proof}

\begin{figure}
  \centering
  \begin{tikzpicture}[quasi]
    \node[point, label=above:{\(x_0\)}, label=left:{\(\mathsf{p_0}\)}] (x1) {};
    \node[point, label=above:{\(x_1\)}, label=below:{\(\mathsf{p_1}\)}, below left=2cm of x1] (x2) {};
    \node[point, label=above:{\(x_2\)}, label=below:{\(\mathsf{p_2}\)},below right=2cm of x1] (x3) {};
    \draw[->]
    (x1) edge (x2)
    (x2) edge (x3)
    (x3) edge (x1)
    ;

    \node[point, label=above:{\(y_0\)}, label=left:{\(\mathsf{p_0}\)}, right=4cm of x1] (y1) {};
    \node[point, label=below:{\(y^\prime_0\)}, label=left:{\(\mathsf{p_0}\)}, below=.75cm of y1] (y1p) {};
    \node[point, label=above:{\(y_1\)}, label=below:{\(\mathsf{p_1}\)}, below left=2cm of y1] (y2) {};
    \node[point, label=above:{\(y_2\)}, label=below:{\(\mathsf{p_2}\)},below right=2cm of y1] (y3) {};
    \draw[->]
    (y1) edge (y2)
    (y2) edge (y3)
    (y3) edge (y1)
    (y1p) edge (y1)
    (y1) edge (y1p)
    (y1p) edge (y2)
    (y3) edge (y1p)
    ;

    \node[below=1.5cm of x1] (m1) {\(\topomodel_1\)}; 
    \node[below=1.5cm of y1] (m1) {\(\topomodel_2\)}; 

  \end{tikzpicture}
  \caption{\(\topomodel_1\) is \(T_0\)-separated, but \(\topomodel_2\) is not.}
  \label{fig:t0-separation}
\end{figure}

 \begin{proposition}
  \label{prop:t1-separation-inexpressible}
  There is no formula of SLCS   expressing \(T_1\) separation.
\end{proposition}
\begin{proof}
  Let \(X\) be an uncountable set. Let \(\mathcal{Y}\) be the set of all subsets of
  \(X\), such that for every \(Y \in \mathcal{Y}\), either \(Y = \emptyset\),
  or the complement of \(Y\) is countable. Then, for every \(x \in X\), let
  \(\nbhd{1}{x} = \{N \mid \exists Y \in \mathcal{Y} \colon Y \subseteq N \land x \in Y\}\).
   Then
  \(\closurespace = (X, \nbhdSing{1})\) is called the \emph{countable complement topology}.
For any valuation \(\valSing{1}\)  over \(X\), 
  \(\topomodel_1 = (\closurespace_1, \indexclspace_{\R}, \valSing{1})\) is a topological model.
  Also, let \(X^\prime\) be constructed from \(X\) by ``doubling'' all points, i.e.,
  \(X^\prime = \{x^\prime \mid x \in X\} \cup X \), where each \(x^\prime\) is a new, distinct, element
  to the \(x\) it is constructed from. Then, let \(\mathcal{Y}^\prime\) be the doubling
  of every set in \(\mathcal{Y}\) in a similar way, and \(\nbhdSing{2}\) be defined similar to \(\nbhdSing{1}\), but over \(\mathcal{Y}^\prime\). Then, \(\closurespace_2 = (X^\prime, \nbhdSing{2})\)
  is the \emph{double pointed countable complement topology}.
  Also, let \(\valSing{2}\) be the valuation that
  assigns the value of \(\val{1}{x}\) to each \(x\)  and \(x^\prime\). Then, \(\topomodel_2 = (\closurespace_2, \indexclspace_{\R}, \valSing{2})\) is also a topological model.

  The relation given by \(x Z_{\nbhdSing{}} y\) iff \(y = x \lor y = x^\prime\) is  obviously
  a neighbourhood  bisimulation. Furthermore, we define \((p,n) Z_1 (q,m)\) iff
  \(p(0) Z_{\nbhdSing{}} q(0)\) and \(p(i) = z \) iff \(q(i) \in \{z, z^\prime\}\), as well
  as \(Z_2 = Z_1^{-1}\). This triple then represents a path preserving bisimulation between
  the two models. 
  However, \(\topomodel_1\) is both \(T_0\) and \(T_1\) separated,
  while \(\topomodel_2\) is neither~\cite{Steen1995}.
\end{proof}

\begin{proposition}
  There is no formula of SLCS  that is expressing
  connectedness.
\end{proposition}
\begin{proof}
  Consider an arbitrary neighbourhood model \(\topomodel\) and
  a model \(\topomodel^\prime\) consisting of two unconnected copies of
  \(\topomodel\). Then we can define a  path preserving neighbourhood
  bisimulation by relating every
  point of \(\topomodel\) with both of its copies in \(\topomodel^\prime\), and
  every path of \(\model\) with both corresponding paths in \(\topomodel^\prime\).
\end{proof}

Similarly,
we can ask whether quasi-discrete models, where the underlying
space is also topological, are only bisimilar to other
models, where the space is topological. As the next lemma shows,
the answer to this question is negative. Hence, SLCS cannot
express transitivity of the underlying edge relation.

\begin{lemma}
  There are quasi-discrete models \(\model_1\) and
  \(\model_2\) that are bisimilar to each other, and where
  the underlying space of \(\model_1\) is topological, while
  the space of \(\model_2\) is not.
\end{lemma}
\begin{proof}
  Consider the graphs in Fig.~\ref{fig:qd-topological}. If we
  set \(x_i Z_{\nbhdSing{}} y_j\) iff \(j \mod 2 = i\), and
  relate paths in the obvious way, then we have a path preserving
  bisimulation. However, \(\model_{\mathit{top}}\) is topological, while
  \(\model_{\mathit{sq}}\) is not.
\end{proof}

\begin{figure}
\centering
\begin{tikzpicture}[quasi]
  \node[ point, label={left:\(x_0\)}, label={right:\(\mathsf{p_0}\)}] (x0) {};
  \node[below = of x0, point, label={left:\(x_1\)}, label={right:\(\mathsf{p_1}\)}] (x1) {};

  \draw[<->]
  (x0) edge (x1)
  ;

  \node[left = of x1, yshift=.5cm] (M1) {\(\topomodel_{\mathit{top}}\)};

  \node[point, right = 3cm of x0, label={left:\(y_0\)}, label={above:\(\mathsf{p_0}\)}] (y0) {};
  \node[point, below = of y0, label={left:\(y_1\)}, label={below:\(\mathsf{p_1}\)}] (y1) {};
  \node[point, right = of y1, label={right:\(y_2\)}, label={below:\(\mathsf{p_0}\)}] (y2) {};
  \node[point, above = of y2, label={right:\(y_3\)}, label={above:\(\mathsf{p_1}\)}] (y3) {};

  \draw[->]
  (y0) edge (y1)
  (y1) edge (y2)
  (y2) edge (y3)
  (y3) edge (y0)
  ;
  \node[right = of y2, yshift=.5cm] (M2) {\(\topomodel_{\mathit{sq}}\)};
\end{tikzpicture}

\caption{Two bisimilar quasi-discrete models, where \(\model_{\mathit{top}}\) is topological and
  \(\model_{\mathit{sq}}\) is not.}
  \label{fig:qd-topological}
\end{figure}

The next example shows that a topological model can be in bisimulation
with non-topological models in a non-trivial way. To that end, we exploit
the transitivity of models being path preserving bisimilar, by
first showing that a specific topological model is path
preserving bisimilar to a topological model with an
underlying quasi-discrete space, and then show that this
second model is path preserving bisimilar to a
model over topological paths, but where the underlying
space is quasi-discrete, but not topological.

\begin{example}
  \label{ex:double-pointed-fully-connected}
  Let \(\model =(\closurespace_2, \indexclspace_\R, \valSing{2})\) be the topological model
  based on the double pointed countable complement topology (cf. the proof of Proposition~\ref{prop:t1-separation-inexpressible}), where \(\val{2}{x} = \{p_0\}\) and \(\val{2}{x^\prime} = \{p_1\}\) for
  any point \(x\) of the underlying set. Furthermore, consider  the models
  depicted in Fig.~\ref{fig:qd-topological}, but considered 
  over the index space \(\indexclspace_{\R}\).
  We will first proceed to define a path preserving bisimulation between
  \(\model\) and \(\model_{top}\).
  
  Let \(x Z_{\nbhdSing{}} x_0 \) and \(x^\prime Z_{\nbhdSing{}} x_1\) for all \(x\) of the underlying
  set of \(\model\). Then clearly \(Z_{\nbhdSing{}}\) is a neighbourhood bisimulation, since
  any neighbourhood in \(\model\) contains both points \(x\) and \(x^\prime\) and
  similarly, any neighbourhood in \(\model_{top}\) contains both \(x_0\) and \(x_1\).

  Now let \(p\) be any path on \(\model\). Then \(q\) defined by
  \(q(i) = x_0\) if \(p(i) \in X\) and \(q(i) = x_1\) if \(p(i) \in X^\prime\),
  is a path as well, since any function into \(\model_{top}\) is continuous
  (as it possesses the indiscrete topology, that is, for both \(x_0\) and \(x_1\),
  \(\{x_0,x_1\}\) is their only neighbourhood). So, we set \((p,m) Z_1 (q,m)\) for
any path, \(m \in \R\) and \(q\) defined as above.
 Hence, whenever there is a
 \(0 < k_q < m\), then \(p(k_q) Z_{\nbhdSing{}} q(k_q)\).

 Finally, consider a path \(q \) on \(\model_{top}\). Choose an arbitrary point
 \(x \in X\), and define \(p\) by
 \(p(i) = x\) if \(q(i) = x_0\) and \(p(i) = x^\prime\) if \(q(i) = x_1\). Then set
 \((q,m) Z_2 (p,m)\) for every \(m \in \R\). Again, the bisimulation condition is satisfied.

 All in all, we have defined a path preserving bisimulation between \(\model\) and
 \(\model_{top}\), where every point of \(\model\) is bisimilar to either
 \(x_0\) or \(x_1\).

 Now we define  a path preserving bisimulation between \(\model_{top}\) and
 \(\model_{sq}\).  As can be easily checked, the relation  
  \( Z_{\nbhdSing{}} = \{(x_0, y_0), (x_0, y_2), (x_1, y_1), (x_1, y_3)\}\) constitutes
  a neighbourhood bisimulation.
  The relation \(Z_2\) can be defined as follows: for any path \(q\) on \(\model_{sq}\) and \(i \in \R\),
  set  \(p(i) = x_0 \) if \(q(i) \in \{y_0,y_2\}\) and \(p(i) = x_1\) otherwise. Then \(p\)
  is continuous, since any function into \(\model_{top}\) is continuous, and also for any index
  \(i\), we have \(p(i) Z_{\nbhdSing{}} q(i)\). Hence, we set \((q,m) Z_2 (p,m)\) for any
  \(m \in \R\).
For  \(Z_1\),
let \(p\) be a path starting in \(x_0\)
and \(m \in \R\). Then we define \(q\) as
\begin{align*}
 q(i) = 
  \begin{cases}
	y_0 & \text{, if } i < 1\\
	y_3 &\text{, if } 1 \leq i < 2\\
	y_2 &\text{, if } 2 \leq i < 3\\
	y_1 &\text{,  if } 3 \leq i
  \end{cases}
\end{align*}
Now, we distinguish several cases:
\begin{enumerate}
\item if \(p(m) = x_0\) and for all \(i < m\), \(p(m) = x_0\), then \((p,m) Z_1 (q,0.5)\),
\item if \(p(m) = x_1\) and for all  \(i < m\), \(p(m) = x_0\), then \((p,m) Z_1 (q,1)\),
\item if \(p(m) = x_0\) and for all  \(i < m\), \(p(m) = x_1\), then \((p,m) Z_1 (q,2)\),
\item if \(p(m) = x_1\) and for all \(i < m\), \(p(m) = x_1\), then \((p,m) Z_1 (q,1.5)\),
\item if \(p(m) = x_0\), for some \(i < m\), \(p(m)= x_0\) and for some \(i< m\), \(p(m) = x_1\), then \((p,m) Z_1 (q,2.5)\), and
\item if \(p(m) = x_1\), for some \(i < m\), \(p(m)= x_0\) and for some \(i < m\), \(p(m)= x_1\), then \((p,m) Z_1 (q,3.5)\).  
\end{enumerate}
For any path with  \(p(0) = x_1\), we can define a path \(q\) in a similar way.
It is easy to check that this relation also satisfies the conditions for a path
preserving bisimulation.
\end{example}


\section{Bisimulations on Quasi-Discrete Spaces}
\label{sec:quasi}
In this section we show how the notions of bisimulation presented in
Sect.~\ref{sec:slcs} relate to common notions of bisimulation for
modal logic when the models taken into considerations are
quasi-discrete neighbourhood models. While being inspired by the
bisimulation of Kurtonina and and de Rijke~\cite{KurtoninaDeRijke97},
we obtain a different result when comparing the path preservering
bisimulation and a bisimulation for modal logic with converse
modalities.

Our notions of bisimulation for quasi-discrete neighbourhood models
are based on the induced edge relation $R_i$ as described in
Sect.~\ref{sec:closurespaces}, and we will refrain in mentioning the
underlying index space to ease the notation. As our first notion of
bisimulation coincides with the standard notion of bisimulation for
modal logic~(e.g.,\cite{BlackburnVanBenthem07}), we refer to it as
modal bisimulation.

\begin{definition}[Modal Bisimulation]
  \label{def:forwardbisim}
  Let \(\topomodel_1 = (X_1, \nbhdSing{1}, \valSing{1})\)
  and \(\topomodel_2 = (X_2, \nbhdSing{2}, \valSing{2})\)
  be two quasi-discrete neighbourhood models.
  A relation \(\bisimrel{} \subseteq X_1 \times X_2\)
  is a \emph{modal bisimulation}, if for every pair
  \(x_1 \bisimrel{} x_2\) the following three conditions hold.
  \begin{description}
  \item[\itatomic] \(p \in \val{1}{x_1}\)
    if, and only if, \(p \in \val{2}{x_2}\)
    for all \(p \in \propositions\)\label{forwbisim:sat}
  \item[\itforthf] if \((x_1,y_1) \in R_1\),
    then there exists \(y_2\in X_2\)
    with \((x_2,y_2) \in R_2\)
    and \(y_1 \bisimrel{} y_2\)\label{forwbisim:forth}
  \item[\itbackf] if \((x_2,y_2) \in R_2\),
    then there exists a \(y_1 \in X_1\)
    with \((x_1,y_1) \in R_1\)
    and \(y_1 \bisimrel{} y_2\)\label{forwbisim:back}
  \end{description}
\end{definition}

Lemma~\ref{lem:modal-nbhd-coincide} shows the relationship between
modal bisimulation and neighbourhood bisimulation on quasi-discrete
neighbourhood models.

\begin{lemma}\label{lem:modal-nbhd-coincide}
  On quasi-discrete neighbourhood models, neighbourhood bisimulation
  and modal bisimulation coincide.
\end{lemma}

In contrast with its behaviour on general neighbourhood spaces,
neighbourhood bisimulation on quasi-discrete neighbourhood models
preserves the ``propagate to'' operator.

\begin{theorem}\label{thm:modal_slcs_wo_reachable_preservation}
  If~$\rho$ is a modal bisimulation between two quasi-discrete
  neighbourhood models $\topomodel_1$ and $\topomodel_2$ with
  $x_1\mathrel{\rho}x_2$, then $\topomodel_1, x_1 \models \varphi$ if,
  and only if, $\topomodel_2, x_2 \models \varphi$ for every
  formula~$\varphi$ of SLCS without $\reachable$.
\end{theorem}

To see that modal bisimulation does not preserve ``reachable from'',
it is enough to consider a very simple example where $\model_1$ is a
model composed of a single point $x$ with valuation
$\valSing{1}(x)=\{p\}$, and $\model_2$ is composed of two points
$\{y_1,y_2\}$ where $\minNbhd{y_1}=\{y_1,y_2\}$,
$\minNbhd{y_2}=\{y_2\}$, $\valSing{2}(y_1)=\{q\}$ and
$\valSing{2}(y_2)=\{p\}$. It is easy to note that $x$ and $y_2$ are
modal bisimilar, but ``reachable from'' is not preserved.  The
preservation of such an operator would require a backward preservation of
paths. This, from a modal logic perspective, corresponds to a notion
of bisimulation able to preserve a modal language with converse
modalities.

\begin{definition}[Modal Bisimulation with Converse]
  \label{def:conversebisim}
  Let
  \(\topomodel_1 = (X_1, \nbhdSing{1},
  \valSing{1})\)
  and
  \(\topomodel_2 = (X_2, \nbhdSing{2},
  \valSing{2})\)
  be two quasi-discrete neighbourhood models.  A relation
  \(\bisimrel{} \subseteq X_1 \times X_2\)
  is a \emph{modal bisimulation with converse}, if it is a modal
  bisimulation and for every pair \(x_1 \bisimrel{} x_2\)
  the following additional conditions hold.
  \begin{description}
  \item[\itforthc] if \((y_1,x_1) \in R_1\),
    then there exists \(y_2\in X_2\)
    with \((y_2,x_2) \in R_2\)
    and \(y_1 \bisimrel{} y_2\) \label{conversebisim:forth}
  \item[\itbackc] if \((y_2,x_2) \in R_2\),
    then there exists a \(y_1 \in X_1\)
    with \((y_1,x_1) \in R_1\)
    and \(y_1 \bisimrel{} y_2\) \label{converseforwbisim:back}
  \end{description}
\end{definition}

\begin{lemma}\label{lem:converse-path-coincide}
  On quasi-discrete neighbourhood models, path preserving bisimulation
  and modal bisimulation with converse coincide.
\end{lemma}

Lemma~\ref{lem:converse-path-coincide} differs from results of
Kurtonina and de Rijke~\cite{KurtoninaDeRijke97}, since their notion
of bisimulation is not equivalent to a bisimulation for temporal
languages preserving simple past and future operators. The reason
being, their semantics for the temporal operator ``since'' and
``until'' has a universal flavour which is not present in our semantic
definition of ``reachable from'' and ``propagate to''.

The following theorem is a direct consequence of
Lemma~\ref{lem:converse-path-coincide} and
Theorem~\ref{thm:slcs-invariant}.

\begin{theorem}\label{thm:converse_slcs_preservation}
  If~$\rho$ is a modal bisimulation with converse between two
  quasi-discrete neighbourhood models $\topomodel_1$ and
  $\topomodel_2$ with $x_1\mathrel{\rho}x_2$, then 
  $\topomodel_1, x_1 \models \varphi$ if, and only if,
  $\topomodel_2, x_2 \models \varphi$ for every formula~$\varphi$ of SLCS.
\end{theorem}


\section{Related Work}
\label{sec:related}

While using logic as a description language for topological
properties has a long tradition, for example in the work of
Tarski \cite{Tarski1938}, only in recent years there has been a resurgence of
spatial interpretations of modal logics. We refer the reader to the survey by Aiello and van Benthem \cite{Aiello2002}, and
the different chapters in the Handbook of Spatial Logics \cite{Aiello2007} for examples of
topological, geometric, and other interpretations.
While the topologic interpretations allow for a topological bisimulation, 
the neighbourhood bisimulation we present in this work is more general, since it is
defined for a larger class of spaces. However, it is straightforward to show that on
topological models (cf. Def.~\ref{def:special-models}), 
topological bisimulation and neighbourhood bisimulation coincide.
A different line of work that is more related to the study of bisimulations is
the spatial logic for concurrency \cite{Caires2001}, which allows for
the structural analysis of pi-calculus processes~\cite{Milner1992}.

Our work directly builds on the definitions of SLCS by Ciancia et al. \cite{Ciancia2017}.
Besides a model checking algorithm for SLCS, they also propose two extensions to the logic.
In the first one, SLCS is extended to incorporate a temporal dimension, which
is treated with different operators than the spatial ones, i.e., the temporal operators from computation tree logic. 
Here, we have instead concentrated solely on the spatial aspects of the language, and leave temporal extensions
of our bisimulations as future work. 
In the second extension, SLCS is equipped with set based modalities, e.g., a
modality \(\mathcal{G} \varphi\) that states the existence of a \emph{path-connected} set \(B\),
such that all elements of \(B\) satisfy \(\varphi\). We intend to examine
this type of modality in the future. Recently, Ciancia et al. investigated
SLCS with coalgebraic methods~\cite{Ciancia2020}. They provide several definitions
of bisimulations, both with and without a coalgebraic flavour,
on quasi-discrete models, show that they coincide, and present an algorithm
and an implementation to minimise a given model with respect to these
bisimulations. Furthermore, they prove that
on the class of quasi-discrete models, where every node has only finitely many pre- and
successors,  logical equivalence is a bisimulation. On general models, however, their analysis
only considers SLCS without path modalities, i.e., the only spatial modality allowed is \emph{near}.
They define a bisimulation, which is similar to definition of neighbourhood bisimulation,
and prove that it coincides with logical equivalence induced by an infinitary modal logic.

The logic STREL of Bartocci et al. \cite{Bartocci2017} is another
extension to SLCS, where the modalities are defined 
to be metric with respect to different
distance functions. That is, for example, they can express that conditions only hold for
paths ``up to three steps'', and similar properties. Therefore, extending our bisimulations
to metric bisimulations in this way is not trivial. In particular, we strongly suspect this would
imply using a kind of metric space as the index space. However, in typical
settings, it is not desirable for the ``metric'' to be symmetric. For example, in
directed graphs, the distance from \(x\) to \(y\) may be different from the other way
around. Such a situation calls for \emph{quasi-metrics}, which only satisfy the triangle
inequality, and that
points of distance zero are identical~\cite{Wilson1931}.

Neighbourhood semantics of modal logics have been studied quite extensively
by now~\cite{Pacuit2017}. However, there are subtle differences to the situation
of our neighourhood models. For one, the logic we study has different
modalities than standard modal logic. In particular, while the \emph{near} modality
is equivalent to the diamond-modality of modal logics with
neighbourhood semantics, the path-based modalities are more expressive.
Furthermore, the spatial interpretation of neighbourhood semantics
is only concerned with topological spaces, while we are considering
the more general notion of arbitrary neighbourhood spaces.


\section{Conclusion}
\label{sec:conclusion}
We have presented path preserving bisimulation, 
a bisimulation on spatial models based on
neighbourhood spaces, a generalisation of topological
spaces. We have then proven that the truth of formulas of the spatial logic SLCS
is preserved between bisimilar points on the models. Using these results,
we have shown that SLCS is not strong enough to express
certain topological properties, such as separation properties or
connectedness. Furthermore, we have compared this bisimulation with
more standard approaches on the subset of purely quasi-discrete models proving that it coincides with modal bisimulation with converse.


There are several natural ways
to extend this line of work.
Up to now,
we have only shown that bisimilarity implies the invariance
of formulas. However, it is important to investigate whether
our bisimulations are matching invariance of formulas exactly,
i.e., whether two points that satisfy the same set
of formulas are also bisimilar. Here, results of Kurtonina
and de Rijke with respect to temporal models might be promising \cite{KurtoninaDeRijke97},
but an adaptation is not straightforward. In particular, they show that
the ultrapower construction of first-order models yields models
that are suitably saturated to contain witnesses of all necessary types.
However, this approach is reliant on the standard translation of modal
logic into first-order logic, a result we do not have at our disposal. This is
due to the second-order nature of the path modalities, which cannot
be reduced to first-order in a similar way as in temporal logic.

It is immediate
that for quasi-discrete models, image-finiteness of the edge relation
means that the minimal neighbourhood of every point is finite.
In this case, the equivalence of points satisfying the same SLCS formulas not using the reachability
modality can easily be proven to be a ``forward path'' preserving bisimulation.
But to treat the full logic SLCS, 
we need
an even stronger notion to obtain a class of models where equivalence of formulas is a
bisimulation.
Even restricting the models such that every point only
possesses finitely many successors \emph{and} predecessors is not
 sufficient. This is due to the fact that \emph{reachable}
quantifies over paths that meet the current point, i.e., in a way we can refer
to ``backwards'' paths, but it is not possible to refer to the
immediate predecessor of a point. To alleviate this, we could introduce a converse modality
to \emph{near}, to distinguish points appropriately.
Ciancia et al. \cite{Ciancia2020} achieved such a
distinction by  employing ``strong'' variants of the reachability modalities, which  allows them
to define such a converse modality as an abbreviation.

Regarding the existing extensions
of SLCS with set-based modalities, we are interested in studying how far
our notion  of bisimulations imply the preservation of such modalities, and whether and how
we would need to strengthen the definitions. A potentially larger
addition would be the investigation of metric variants of SLCS~\cite{Bartocci2017},
and what kind of metrics or generalised metrics are appropriate
in this case.


\bibliography{lit}

\begin{thebibliography}{10}

\bibitem{Aiello2007}
Marco Aiello, Ian Pratt-Hartmann, and Johan~van Benthem, editors.
\newblock {\em Handbook of {Spatial} {Logics}}.
\newblock Springer, 2007.

\bibitem{Aiello2002}
Marco Aiello and Johan van Benthem.
\newblock A modal walk through space.
\newblock {\em Journal of Applied Non-Classical Logics}, 12(3-4):319--363,
  2002.
\newblock \href {https://doi.org/10.3166/jancl.12.319-363}
  {\path{doi:10.3166/jancl.12.319-363}}.

\bibitem{Banci2019}
Fabrizio Banci~Buonamici, Gina Belmonte, Vincenzo Ciancia, Diego Latella, and
  Mieke Massink.
\newblock Spatial logics and model checking for medical imaging.
\newblock {\em International Journal on Software Tools for Technology
  Transfer}, February 2019.
\newblock \href {https://doi.org/10.1007/s10009-019-00511-9}
  {\path{doi:10.1007/s10009-019-00511-9}}.

\bibitem{Bartocci2017}
Ezio Bartocci, Luca Bortolussi, Michele Loreti, and Laura Nenzi.
\newblock Monitoring {Mobile} and {Spatially} {Distributed} {Cyber}-physical
  {Systems}.
\newblock In {\em Proceedings of the 15th {ACM}-{IEEE} {International}
  {Conference} on {Formal} {Methods} and {Models} for {System} {Design}},
  {MEMOCODE} '17, pages 146--155, New York, NY, USA, 2017. ACM.
\newblock event-place: Vienna, Austria.
\newblock \href {https://doi.org/10.1145/3127041.3127050}
  {\path{doi:10.1145/3127041.3127050}}.

\bibitem{Baryshnikov2009}
Yuliy Baryshnikov and Robert Ghrist.
\newblock Target enumeration via {E}uler characteristic integrals.
\newblock {\em SIAM Journal on Applied Mathematics}, 70(3):825--844, 2009.
\newblock \href {https://doi.org/10.1137/070687293}
  {\path{doi:10.1137/070687293}}.

\bibitem{DBLP:conf/tacas/BelmonteCLM19}
Gina Belmonte, Vincenzo Ciancia, Diego Latella, and Mieke Massink.
\newblock {VoxLogicA}: {A} spatial model checker for declarative image
  analysis.
\newblock In Tom{\'{a}}s Vojnar and Lijun Zhang, editors, {\em Tools and
  Algorithms for the Construction and Analysis of Systems - 25th International
  Conference, {TACAS} 2019, Held as Part of the European Joint Conferences on
  Theory and Practice of Software, {ETAPS} 2019, Prague, Czech Republic, April
  6-11, 2019, Proceedings, Part {I}}, volume 11427 of {\em Lecture Notes in
  Computer Science}, pages 281--298. Springer, 2019.
\newblock \href {https://doi.org/10.1007/978-3-030-17462-0\_16}
  {\path{doi:10.1007/978-3-030-17462-0\_16}}.

\bibitem{BlackburnVanBenthem07}
Patrick Blackburn and Johan van Benthem.
\newblock Modal logic: a semantic perspective.
\newblock In {\em Handbook of Modal Logic}, pages 1--84. North-Holland, 2007.
\newblock \href {https://doi.org/10.1016/s1570-2464(07)80004-8}
  {\path{doi:10.1016/s1570-2464(07)80004-8}}.

\bibitem{Caires2001}
Lu\'{i}s Caires and Luca Cardelli.
\newblock A spatial logic for concurrency ({Part} {I}).
\newblock In N.~Kobayashi and B.~C. Pierce, editors, {\em International
  {Symposium} on {Theoretical} {Aspects} of {Computer} {Software} -- {TACS}
  2001}, volume 2215 of {\em {LNCS}}, pages 1--37. Springer, 2001.
\newblock \href {https://doi.org/10.1007/3-540-45500-0\_1}
  {\path{doi:10.1007/3-540-45500-0\_1}}.

\bibitem{Ciancia2018}
Vincenzo Ciancia, Stephen Gilmore, Gianluca Grilletti, Diego Latella, Michele
  Loreti, and Mieke Massink.
\newblock Spatio-temporal model checking of vehicular movement in public
  transport systems.
\newblock {\em International Journal on Software Tools for Technology
  Transfer}, 20:289--311, January 2018.
\newblock \href {https://doi.org/10.1007/s10009-018-0483-8}
  {\path{doi:10.1007/s10009-018-0483-8}}.

\bibitem{Ciancia2017}
Vincenzo Ciancia, Diego Latella, Michele Loreti, and Mieke Massink.
\newblock {Model Checking Spatial Logics for Closure Spaces}.
\newblock {\em {Logical Methods in Computer Science}}, {Volume 12, Issue 4},
  April 2017.
\newblock \href {https://doi.org/10.2168/LMCS-12(4:2)2016}
  {\path{doi:10.2168/LMCS-12(4:2)2016}}.

\bibitem{Ciancia2020}
Vincenzo Ciancia, Diego Latella, Mieke Massink, and Erik de~Vink.
\newblock Towards spatial bisimilarity for closure models: Logical and
  coalgebraic characterisations, 2020.
\newblock \href {http://arxiv.org/abs/2005.05578} {\path{arXiv:2005.05578}}.

\bibitem{Ciancia2016}
Vincenzo Ciancia, Diego Latella, Mieke Massink, Rytis Pa\v{s}kauskas, and
  Andrea Vandin.
\newblock A {Tool}-{Chain} for {Statistical} {Spatio}-{Temporal} {Model}
  {Checking} of {Bike} {Sharing} {Systems}.
\newblock In Tiziana Margaria and Bernhard Steffen, editors, {\em Leveraging
  {Applications} of {Formal} {Methods}, {Verification} and {Validation}:
  {Foundational} {Techniques}}, number 9952 in Lecture {Notes} in {Computer}
  {Science}, pages 657--673. Springer International Publishing, October 2016.
\newblock \href {https://doi.org/10.1007/978-3-319-47166-2_46}
  {\path{doi:10.1007/978-3-319-47166-2_46}}.

\bibitem{Emerson1986}
E.~Allen Emerson and Joseph~Y. Halpern.
\newblock "{Sometimes}" and "{Not} {Never}" {Revisited}: {On} {Branching}
  {Versus} {Linear} {Time} {Temporal} {Logic}.
\newblock {\em J. ACM}, 33(1):151--178, January 1986.
\newblock \href {https://doi.org/10.1145/4904.4999}
  {\path{doi:10.1145/4904.4999}}.

\bibitem{Galton2003}
Antony Galton.
\newblock A generalized topological view of motion in discrete space.
\newblock {\em Theoretical Computer Science}, 305(1):111 -- 134, 2003.
\newblock \href {https://doi.org/10.1016/S0304-3975(02)00701-6}
  {\path{doi:10.1016/S0304-3975(02)00701-6}}.

\bibitem{KurtoninaDeRijke97}
Natasha Kurtonina and Maarten de~Rijke.
\newblock Bisimulations for temporal logic.
\newblock {\em Journal of Logic, Language and Information}, 6(4):403--425,
  1997.
\newblock \href {https://doi.org/10.1023/A:1008223921944}
  {\path{doi:10.1023/A:1008223921944}}.

\bibitem{doi:10.1098/rspa.2019.0278}
Sven Linker and Michele Sevegnani.
\newblock Target counting with presburger constraints and its application in
  sensor networks.
\newblock {\em Proceedings of the Royal Society A: Mathematical, Physical and
  Engineering Sciences}, 475(2231):20190278, 2019.
\newblock \href {https://doi.org/10.1098/rspa.2019.0278}
  {\path{doi:10.1098/rspa.2019.0278}}.

\bibitem{Milner1992}
Robin Milner, Joachin Parrow, and David Walker.
\newblock A {Calculus} of {Mobile} {Processes}, {I}.
\newblock {\em Information and Computation}, 100(1):1--40, September 1992.
\newblock \href {https://doi.org/10.1016/0890-5401(92)90008-4}
  {\path{doi:10.1016/0890-5401(92)90008-4}}.

\bibitem{Pacuit2017}
Eric Pacuit.
\newblock {\em Neighborhood {Semantics} for {Modal} {Logic}}.
\newblock Springer, Cham, 2017.
\newblock \href {https://doi.org/10.1007/978-3-319-67149-9}
  {\path{doi:10.1007/978-3-319-67149-9}}.

\bibitem{8064025}
Danilo Pianini, Simon Dobson, and Mirko Viroli.
\newblock Self-stabilising target counting in wireless sensor networks using
  euler integration.
\newblock In {\em 2017 IEEE 11th International Conference on Self-Adaptive and
  Self-Organizing Systems (SASO)}, pages 11--20, Sept 2017.
\newblock \href {https://doi.org/10.1109/SASO.2017.10}
  {\path{doi:10.1109/SASO.2017.10}}.

\bibitem{Pnueli1977}
Amir Pnueli.
\newblock The {Temporal} {Logic} of {Programs}.
\newblock In {\em {IEEE} {Symposium} on {Foundations} of {Computer} {Science}
  -- {SFCS} 1977}, pages 46--57. IEEE Computer Society, 1977.
\newblock \href {https://doi.org/10.1109/SFCS.1977.32}
  {\path{doi:10.1109/SFCS.1977.32}}.

\bibitem{Steen1995}
Lynn~Arthur Steen and {J. Arthur Seebach, Jr.}
\newblock {\em Counterexamples in Topology}.
\newblock Springer-Verlag, New York, 1978.
\newblock Reprinted by Dover Publications, New York, 1995.

\bibitem{Tarski1938}
Alfred Tarski.
\newblock Der {A}ussagenkalk\"{u}l und die {T}opologie.
\newblock {\em Fundamenta Mathematicae}, 31:103--134, 1938.
\newblock \href {https://doi.org/10.4064/fm-31-1-103-134}
  {\path{doi:10.4064/fm-31-1-103-134}}.

\bibitem{Cech1966}
Eduard \v{C}ech, Zden\v{e}k Frol\'{i}k, and Miroslav Kat\v{e}tov.
\newblock {\em Topological spaces}.
\newblock Academia, Publishing House of the Czechoslovak Academy of Sciences,
  1966.
\newblock URL: \url{http://eudml.org/doc/277000}.

\bibitem{Wilson1931}
Wallace~Alvin Wilson.
\newblock On quasi-metric spaces.
\newblock {\em American Journal of Mathematics}, 53(3):675--684, 1931.
\newblock \href {https://doi.org/10.2307/2371174} {\path{doi:10.2307/2371174}}.

\end{thebibliography}

\newpage{}
\appendix

\section{Proofs of Section~\ref{sec:slcs}}
\label{sec:general_proofs}

\newcounter{savedlemma}

\setcounter{savedlemma}{\thelemma}
\setcounter{lemma}{15}
\begin{lemma}[restated]
    Let \((Z_{\nbhdSing{}}, Z_1, Z_2)\) be a path preserving bisimulation between \(\model_1\) and
  \(\model_2\), and \(\varphi\) be an SLCS formula that is invariant under
  neighbourhood bisimulation, i.e., for any \(x_1 \in \model_1\) and \(x_2 \in \model_2\)  
  with \(x_1 Z_{\nbhdSing{}} x_2\), we have
  \(\model_1, x_1 \models \varphi\) if, and only if, \(\model_2, x_2 \models \varphi\).
  For two paths \(p\) and \(q\) with \((p,n)\mathrel{Z_1} (q,m)\), we have
  \(\model_1, p,(0,n) \models \varphi\) implies \(\model_2,q,(0,m) \models \varphi\). Additionally,
  if \((q,m) Z_2 (p,n)\) then \(\model_2, q,(0,m) \models \varphi\) implies \(\model_1, p,(0,n) \models
  \varphi\).
  
\end{lemma}
\begin{proof}
  Assume \(\model_1, p,(0,n) \models \varphi\) and  \((p,n)\mathrel{Z_1} (q,m)\), and
  let \(k_q\) be an arbitrary index such that \(0 < k_q < m\).
  We need to show that \(\model_2, q(k_q)\models \varphi\). By the bisimulation
  property (Def.~\ref{def:path_preserving_bisim}~(\ref{ppb:ex_path_forth2})), we know that there is a \(k_p\) such that \(0 < k_p < n\) and \(p(k_p) Z_{\nbhdSing{}} q(k_q)\). By the semantics of path intervals, we have \(\model_1, p(k_p) \models \varphi\), and
  since \(\varphi\) is invariant under neighbourhood bisimulation, we get \(\model_2, q(k_q) \models
  \varphi\). Since \(k_q\) was arbitrary, we have \(\model_2, q, (0,m) \models \varphi\). The
  other case is similar.
\end{proof}
\setcounter{lemma}{\thesavedlemma}

\setcounter{savedlemma}{\thelemma}
\setcounter{lemma}{16}
\begin{theorem}[restated]
    If \((Z_{\nbhdSing{}},Z_1,Z_2)\) is a path preserving bisimulation between
  \(\model_1\) and \(\model_2\)
  with \(x_1 Z_{\nbhdSing{}} x_2\), then \(\model_1, x_1 \models \varphi\) if, and only if,
  \(\model_2, x_2 \models \varphi\) for every formula \(\varphi\) of SLCS.
\end{theorem}
\begin{proof}
    We proceed
  by induction on the length of formulas. The induction base and the cases for the Boolean
  operators are as usual.

  So consider \(\topomodel_1, x_1 \models \near \varphi\). That is, \(x_1 \in \closure{1}{\{y \mid \topomodel_1, y \models \varphi\}}\), which by Def.~\ref{def:nbhd-space} is equivalent
  to \(x_1 \in \{ z \mid \forall N \in \nbhd{1}{z} \colon N \cap \{y \mid \topomodel_1, y \models \varphi\} \neq \emptyset\}\). Hence \(\forall N \in \nbhd{1}{x_1} \colon \exists y \in N \colon \topomodel_1, y \models \varphi\). Now choose an arbitrary neighbourhood \(N_2\) of \(x_2\), i.e., \(N_2 \in \nbhd{2}{x_2}\). By condition \(\itforthnbhd\) of Def.~\ref{def:nbhd_bisim}, there is a neighbourhood \(N_1 \in \nbhd{1}{x_1}\) such that for all \(y_1 \in N_1\), there is a \(y_2 \in N_2\) with \(y_1 Z_{\nbhdSing{}} y_2\). In
  particular, this is the case for the \(y_1\) with \(\topomodel_1, y_1\models \varphi\). Hence,
  by the induction hypothesis, \(\topomodel_2, y_2 \models \varphi\). Since \(N_2\) was arbitrary,
  we have \(\forall N_2 \in \nbhd{2}{x_2} \colon \exists y \in N_2 \colon \topomodel_2,y \models \varphi\). That is, \(x_2 \in  \{ z \mid \forall N \in \nbhd{2}{z} \colon N \cap \{y \mid \topomodel_2, y \models \varphi\} \neq \emptyset\} = \closure{2}{\{y \mid \topomodel_2, y \models \varphi\}}\). Hence, \(\topomodel_2, x_2 \models \near \varphi\). The other direction is similar.

  Now let \(\model_1,x_1 \models \varphi \propagate \psi\). That is, there is a path
  \(p\) with \(p(0) = x_1\) and an \(n\) such that \(\model_1, p(n) \models \psi\),
  \(\model_1, x_1 \models \varphi\) and \(\model_1, p, (0,n) \models \varphi\). Now, by
  the induction hypothesis, we have \(\model_2, x_2 \models \varphi\). Furthermore,
  by  Def.~\ref{def:path_preserving_bisim}, there is a path \(q\) on \(\model_2\)
  with \(q(0) = x_2\) and \(m\) such that \((p,n) Z_1 (q,m)\) and \(p(n) Z_{\nbhdSing{}} q(m)\).
    Hence, \(\model_2, q(m) \models \psi\), and by Lemma~\ref{lem:auxilliary_path_preservation},
    we have \(\model_2, q,(0,m) \models \varphi\). All in all, \(\model_2, x_2 \models \varphi \propagate \psi\). The other direction is similar, using \(Z_2\) and the other case of Lemma~\ref{lem:auxilliary_path_preservation}. 

    The case for \( \varphi \reachable \psi\) is similar to the
    preceeding case, using the additional cases in Def.~\ref{def:path_preserving_bisim} as indicated
    in the last item. For illustration, we prove the first subcase. So assume
    \(\model_1, x_1 \models \varphi \reachable \psi\). Hence, there is a path \(p\) on
    \(\model_1\) and an \(n\) such that \(p(n) = x_1\) and \(\model_1, p(0) \models \psi\),
    \(\model_1, x_1 \models \varphi\) and \(\model_1, p,(0,n) \models \varphi\). By
    Def.~\ref{def:path_preserving_bisim}, we then have that there is a
    path \(q\) on \(\model_2\) and an \(m\) such that \((p,n) Z_1 (q,m)\) and
    \(p(0) Z_{\nbhdSing{}} q(0)\). By the induction hypothesis, we get \(\model_2, q(m) \models
    \varphi\), \(\model_2, q(0) \models \psi\), and then, by
    Lemma~\ref{lem:auxilliary_path_preservation}, we also have \(\model_2, q, (0,m) \models \varphi\).
    Hence, \(\model_2, x_2 \models \varphi \reachable \psi\). 
\end{proof}
\setcounter{lemma}{\thesavedlemma}

\section{Proofs of Section~\ref{sec:quasi}}
\label{sec:q-d_proofs}

Proofs in this section rely on definitions of modal bisimulation based
on the notion of minimal neighbourhood. This is possible due to the
strong relationship between the edge relation and the minimal
neighbourhood. In particular, the definition of modal bisimulation can
be rewritten in terms of minimal neighbourhood, as \itforthf (resp.,
\itbackf) can be rewritten as for every $y_1 \in \minNbhd{x_1}$
(resp., $y_2 \in \minNbhd{x_2}$) there exists $y_2 \in \minNbhd{x_2}$
(resp., $y_1 \in \minNbhd{x_1}$) and $y_1 \bisimrel{}
y_2$.
Analogously, the definition of modal bisimulation with converse can be
rewritten in terms of minimal neighbourhood, as \itforthc (resp.,
\itbackc) can be rewritten as for every
$y_1 \in \{y\in X_1 \mid x_1\in \minNbhd{y}\} = \closureSing{}(x_1)$
(resp.,
$y_2 \in \{y\in X_2 \mid x_2\in \minNbhd{y}\} = \closureSing{}(x_2)$)
there exists $y_2 \in \closureSing{}(x_2)$ (resp.,
$y_1 \in \closureSing{}(x_1)$) and $y_1 \bisimrel{} y_2$.

\setcounter{savedlemma}{\thelemma}
\setcounter{lemma}{24}
\begin{lemma}[restated]
  On quasi-discrete neighbourhood models, neighbourhood bisimulation
  and modal bisimulation coincide.
\end{lemma}

\begin{proof}
  Let $\topomodel_1 = (X_1, \nbhdSing{1}, \valSing{1})$ and
  $\topomodel_2 = (X_2, \nbhdSing{2},\valSing{2})$ be two
  quasi-discrete neighbourhood models, and
  $\bisimrel{} \subseteq X_1 \times X_2$ a relation between them. We
  show that $\bisimrel{}$ is a modal bisimulation iff it is a
  neighbourhood bisimulation.

  $(\Rightarrow)$ Assume $x_1 \bisimrel{} x_2$. Atomic equivalence is
  trivially true. By \itforthf for any $y_1 \in \minNbhd{x_1}$ there
  exists $y_2\in \minNbhd{x_2}$ with $y_1\bisimrel{} y_2$. As
  $\minNbhd{x_2} \subseteq N$ for any $N \in \nbhdSing{2}(x_2)$, it is
  always possible to chose $\minNbhd{x_1}$ to satisfy the \itforthnbhd
  condition. Hence, on quasi-discrete neighbourhood models
  \itforthf implies \itforthnbhd. The backward direction is analogous.

  $(\Leftarrow)$ Assume $x_1 \bisimrel{}x_2$. Atomic equivalence is
  trivially true. By \itforthnbhd for $\minNbhd{x_2}$ there exists a
  neighbourhood $N_1 \in \nbhdSing{1}(x_1)$ such that for every
  $y_1 \in N_1$ there exists $y_2\in \minNbhd{x_2}$ with
  $y_1 \bisimrel{} y_2$. As $\minNbhd{x_1} \subseteq N_1$, it follows
  that on quasi-discrete neighbourhood models, \itforthnbhd implies
  \itforthf. The backward direction is analogous.
\end{proof}
\setcounter{lemma}{\thesavedlemma}

In order to prove
Theorem~\ref{thm:modal_slcs_wo_reachable_preservation}, we first show a
stronger result on preservation of paths.

\begin{lemma}\label{lem:quasi_forward_path_preservation}
  If~$\rho$ is a modal bisimulation between two quasi-discrete
  neighbourhood models $\topomodel_1$ and $\topomodel_2$ with
  $x_1\mathrel{\rho}x_2$, then for every path
  $p \colon x_1 \rightsquigarrow \infty$ there exists a path
  $q \colon x_2 \rightsquigarrow \infty$ such that for any $n\in \N$
  it holds that $p(n)\mathrel{\rho} q(n)$, and the other way around.
\end{lemma}

\begin{proof}
  We recursively build the path $q$ as follows. First, set
  $q(0) = x_2$.  Second, if $q(k)$ is defined and
  $p(k)\mathrel{\rho}q(k)$, then by modal bisimulation there exists
  some $y \in \minNbhd{q(k)}$ with $p(k+1)\mathrel{\rho}y$, and we set
  $q(k+1) = y$. By construction we have that $p(n)\mathrel{\rho} q(n)$
  for any $n\in\N$. We need to show that $q$ is a continuous
  function. For quasi-discrete neighbour models this means to show
  that for any $\{n, n+1\}$ we have that
  $q[\{n,n+1\}]\subseteq \minNbhd{q(n)}$, which follows by
  construction.
\end{proof}

\setcounter{savedlemma}{\thelemma}
\setcounter{lemma}{25}
\begin{theorem}[restated]
  If~$\rho$ is a modal bisimulation between two quasi-discrete
  neighbourhood models $\topomodel_1$ and $\topomodel_2$ with
  $x_1\mathrel{\rho}x_2$, then $\topomodel_1, x_1 \models \varphi$ if,
  and only if, $\topomodel_2, x_2 \models \varphi$ for every
  formula~$\varphi$ of SLCS without $\reachable$.
\end{theorem}

\begin{proof}
  We proceed by induction on the length of formulas. The induction
  base and the cases for the Boolean operators are as usual.

  Consider \(\topomodel_1, x_1 \models \near \varphi\).
  On quasi-discrete neighbourhood models this means that there exists
  $x_1'\in \minNbhd{x_1}$ such that
  $\topomodel_1, x_1' \models \varphi$. By \itforthf, there exists
  $x_2' \in \minNbhd{x_2}$ such that $x_1'\mathrel{\rho}x_2'$ and, by
  IH, $\topomodel_2, x_2' \models \varphi$. Hence,
  $\topomodel_1, x_2 \models \near \varphi$. The other direction is
  similar.

  Consider \(\model_1,x_1 \models \varphi \propagate \psi\).
  That is, there is a path \(p\)
  and an $n$ such that \(p(0) = x_1\)  and  $\model_1, p(i) \models \varphi$ for all
  $0\leq i < n$, and $\model_1, p(n)\models \psi$. By
  Lemma~\ref{lem:quasi_forward_path_preservation} there exists a path
  $q$ on $\model_2$ with $q(0) = x_2$, and such that
  $p(i)\mathrel{\rho}q(i)$ for all $i\in\N$. Then by IH,
  $\model_2, q(i) \models \varphi$ for all $0\leq i < n$, and
  $\model_2,q(n)\models \psi$. Hence,
  $\model_2,x_2\models \varphi \propagate \psi$ The other direction is
  similar.
\end{proof}
\setcounter{lemma}{\thesavedlemma}

In order to prove
Lemma~\ref{lem:converse-path-coincide}, we first show a
stronger result on preservation of paths.

\begin{lemma}\label{lem:quasi_path_preservation}
  If~$\rho$ is a modal bisimulation with converse between two quasi-discrete
  neighbourhood models $\topomodel_1$ and $\topomodel_2$ with
  $x_1\mathrel{\rho}x_2$, then for every path
  $p \colon y_1 \rightsquigarrow \infty$ with $p(n) = x_1$ there
  exists a path $q \colon y_2 \rightsquigarrow \infty$ with
  $q(n) = x_2$ such that for any $i\in \N$ it holds that
  $p(i)\mathrel{\rho} q(i)$, and the other way around.
\end{lemma}

\begin{proof}
  We recursively build the path $q$ as follows. First we set
  $q(n) = x_2$, and all $q(i)$ values with $i\geq n$ are defined as in
  Lemma~\ref{lem:quasi_forward_path_preservation}. Second, if $q(k)$
  with $0 < k \leq n$ is defined and $p(k)\mathrel{\rho}q(k)$, then by
  modal bisimulation with converse there exists some $y$ with
  $q(k) \in \minNbhd{y}$ and $p(k-1)\mathrel{\rho}y$, and we set
  $q(k-1) = y$. By construction we have that $p(i)\mathrel{\rho} q(i)$
  for any $i\in\N$, and continuity of $q$ is as in
  Lemma~\ref{lem:quasi_forward_path_preservation}.
\end{proof}

\setcounter{savedlemma}{\thelemma}
\setcounter{lemma}{27}
\begin{lemma}[restated]
  On quasi-discrete neighbourhood models, path preserving bisimulation
  and modal bisimulation with converse coincide.
\end{lemma}

\begin{proof}  
  Let $\model_1$ and $\model_2$ be two quasi-discrete neighbourhood
  models. To prove the lemma, we show that (1) if
  $(Z_{\nbhdSing{}}, Z_1, Z_2)$ is a path preserving bisimulation
  between $\model_1$ and $\model_2$, then $Z_{\nbhdSing{}}$ is a modal
  bisimulation with converse; and (2) if $\rho$ is a modal
  bisimulation with converse, $\rho$ induces a path preserving
  bisimulation $(\rho, Z_1, Z_2)$.

  (1). Assume $x_1 \mathrel{Z_{\nbhdSing{}}} x_2$. Atomic equivalence
  is trivially true. By point 2 of
  Definition~\ref{def:path_preserving_bisim} for any path
  $p \colon x_1 \rightsquigarrow \infty$ and $n\neq 0$ there exists
  $q \colon x_2 \rightsquigarrow \infty$ and $m$
  s.t. $p(n)\mathrel{Z_{\nbhdSing{}}}q(m)$. On quasi-discrete
  neighbourhood models, if $n=1$, then $m=1$ and we have that for any
  $y_1 \in \minNbhd{x_1}$ there exists $y_2 \in \minNbhd{x_2}$
  s.t. $y_1\mathrel{Z_{\nbhdSing{}}}y_2$. Hence, $Z_{\nbhdSing{}}$
  satisfies \itforthf. The direction for \itbackf is analogous by
  point 5 of Definition~\ref{def:path_preserving_bisim}, and a similar
  argument also holds for \itforthc and \itbackc by points 3 and 6 of
  Definition~\ref{def:path_preserving_bisim}.

  (2). Assume $x_1 \bisimrel{}x_2$ and let
  $p \colon x_1 \rightsquigarrow \infty$ be a path starting from
  $x_1$. By Lemma~\ref{lem:quasi_forward_path_preservation} there
  exists a path $q \colon x_2 \rightsquigarrow \infty$
  s.t. $p(i)\mathrel{\rho}q(i)$ for all $i\in\N$. Let us set
  $(p,n)\mathrel{Z_1^{p,n}}(q,n)$ and $Z_2^{q,n}$ as the inverse of
  $Z_1^{p,n}$. Let $Z_1 = \bigcup_{p \in \mathcal P, i>0} Z_1^{p,i}$
  and $Z_2 = \bigcup_{q \in \mathcal Q, i>0} Z_2^{q,i}$ with
  $\mathcal P$ (resp., $\mathcal Q$) the set of paths over $\model_1$
  (resp., $\model_2$) starting from bisimilar points. It is immediate
  that $(\rho, Z_1, Z_2)$ satisfies points 2, 4, 5 and 7 of
  Definition~\ref{def:path_preserving_bisim}. The cases for points 3
  and 6 of Definition~\ref{def:path_preserving_bisim} is analogous by
  using paths defined in the proof of
  Lemma~\ref{lem:quasi_path_preservation}. Hence, $\rho$ induces a
  path preserving bisimulation $(\rho, Z_1, Z_2)$.
\end{proof}
\setcounter{lemma}{\thesavedlemma}


\end{document}